\documentclass[12pt]{article}
\usepackage{graphicx}
\usepackage{amsmath,amsthm,amssymb,amscd,cite,color,enumerate}
\theoremstyle{plain}
\newtheorem{theorem}{Theorem}[section]
\newtheorem{proposition}[theorem]{Proposition}
\newtheorem{lemma}[theorem]{Lemma}
\newtheorem{corollary}[theorem]{Corollary}

\newtheorem{conjecture}{Conjecture}[section]

\theoremstyle{definition}
 \newtheorem{example}{Example}[section]
 \newtheorem{remark}{Remark}[section]

\renewenvironment{proof}[1][\proofname]{    {\noindent \textit{Proof.}}
}{{\hfill{$\square$}}}

\newcommand{\F}{\mathbb{F}}

\def\rank{{\rm rank}}

\begin{document}

\title{Linear $\ell$-Intersection  Pairs of Codes and Their Applications}

\author{Kenza Guenda, T.~Aaron Gulliver,
        Somphong  Jitman\\ and Satanan Thipworawimon\thanks{K. Guenda and
        T. A. Gulliver  are with the  Department of Electrical and Computer Engineering, University of
        Victoria, PO Box 1700, STN CSC, Victoria, BC, Canada V8W 2Y2.
        Email: {kguenda@uvic.ca, agullive@ece.uvic.ca}.
        S. Jitman  and S. Thipworawimon are with the  Department of Mathematics, Faculty of Science, Silpakorn University,
        Nakhon Pathom 73000, Thailand.
 Email: {sjitman@gmail.com, thipworawimon\_s@su.ac.th}.
          }}

\maketitle

\begin{abstract}
   In this paper,  a linear  $\ell$-intersection pair of codes  is introduced as  a generalization of linear complementary pairs  of codes.
Two linear codes are said to be a linear $\ell$-intersection pair if their intersection has dimension $\ell$.
Characterizations and constructions  of such pairs of codes are given in terms of the corresponding generator and parity-check matrices.
Linear  $\ell$-intersection pairs of MDS codes over $\mathbb{F}_q$ of length up  to $q+1$ are given for all possible parameters.
As an application, linear $\ell$-intersection pairs of codes  are used to construct entanglement-assisted quantum error correcting codes.
This provides a large number of new MDS entanglement-assisted quantum error correcting codes.

\end{abstract}



 \section{Introduction}

Linear complementary pairs (LCPs) of codes have been of interest and extensively studied due to their rich algebraic structure and wide applications in cryptography.
For example, in \cite{carlet1} and \cite{carlet2},  it was shown that these pairs of codes can be used to counter passive and active side-channel analysis attacks on embedded cryptosystems.
Several construction of LCPs of codes were also given.

In this paper,   we introduce a linear  $\ell$-intersection pair of codes as  a generalization of the LCP of codes in \cite{carlet2}.
Two linear codes are said to be a linear  $\ell$-intersection pair if their intersection has dimension $\ell$.
A characterization of such pairs of codes is given in terms of  generator and parity-check matrices of codes.
We construct $\ell$-intersection pairs from generalized Reed-Solomon and extended  generalized Reed-Solomon codes.
Linear  $\ell$-intersection pairs of maximum distance separable (MDS) codes over $\mathbb{F}_q$ of length up  to $q+1$ are given for all possible parameters.
As an application, linear $\ell$-intersection pairs of codes  are employed to construct  entanglement-assisted quantum error correcting codes.
This gives a large family of  new  MDS  entanglement-assisted quantum error correcting codes.

The remainder of this paper is organized as follows.
Some properties, characterizations and constructions of linear $\ell$-intersection pairs of codes are given in  Section~\ref{sect2}.
In Section~\ref{sectionMDS}, constructions  of linear $\ell$-intersection pairs of MDS codes are discussed.
The application of linear  $\ell$-intersection pairs in the construction of entanglement-assisted quantum error correcting codes
is considered in Section~\ref{sectApp}.
Some concluding remarks are given in Section~\ref{secConclud} along with some open problems.

\section{$\ell$-Intersection Pairs of Codes} \label{sect2}

Let $\F_q$ denote the finite field of $q$ elements where $q$ is a prime power.
For positive integers $k\leq n$ and $d$,
an $[n,k,d]_q $ {\em linear code} is defined to be a $k$-dimensional subspace of $\F_q^{n}$ with minimum Hamming distance $d$.
An $[n,k,d]_q$ code is called {\em maximum distance separable} (MDS) if the parameters satisfy $d=n-k+1$.
The {Euclidean inner product} between $\boldsymbol{u}=(u_1,u_2,\dots, u_n)$ and $\boldsymbol{v}=(v_1,v_2,\dots, v_n)$  in $\mathbb{F}_q$ is defined as
\[
\langle \boldsymbol{u}, \boldsymbol{v} \rangle=\sum _{i=1}^n u_iv_i .
\]
The {\em (Euclidean) dual}  $C^{\perp}$ of a linear code $C$ is defined as
\[
C^{\perp}=\{  \boldsymbol{u} \in \mathbb{F}_q^n \mid \langle \boldsymbol{u}, \boldsymbol{c} \rangle =0 \text{ for all } \in C\}.
\]
A linear code  $C$ is called {\em linear complementary dual} (LCD) if $C\cap C^\perp=\{\boldsymbol{0}\}$.  Linear codes  $C$ and $D$  of length $n$ over $\mathbb{F}_q$ are  called a {\em linear complementary pair} (LCP) if  $C\cap D =\{\boldsymbol{0}\} $ and $C+D=\mathbb{F}_q^n$.

For an integer $\ell\geq 0$, linear codes  $C$ and $D$  of length $n$ over $\mathbb{F}_q$ are  called a {\em linear $\ell$-intersection pair} if  $\dim(C\cap D) =\ell$.
From the definition above, we have the following observations.
\begin{itemize}
    \item A linear $0$-intersection pair with $\dim(C)+\dim(D)=n$ is an LCP (see \cite{carlet2}).
    \item A linear $0$-intersection pair with $D=C^\perp$ is an LCD code (see \cite{M1992}).
    \item The study of a linear  $\ell$-intersection pair  with $D=C^\perp$ is equivalent to that of the hull $Hull(C):=C\cap C^\perp$ of $C$ (see \cite{GJG}).
\end{itemize}
Therefore, linear $\ell$-intersection pairs of codes can be viewed as a generalization of LCPs of codes, LCD codes, and the hulls of codes.

\subsection{Characterizations of Linear $\ell$-Intersection Pairs of Codes}\label{ssec2.1}
In this section,  properties of linear $\ell$-intersection pairs of codes are established in terms of their generator and parity-check matrices.
In some cases, links between this concept and known families of codes such as complementary dual codes, self-orthogonal codes, and linear complementary pairs of codes,
as well as hulls of codes, are discussed.

\begin{theorem}  \label{thmGLCP}
    For $i\in\{1,2\}$, let $C_i$ be a linear $[n,k_i]_{q}$ code with parity check matrix $H_i$ and generator matrix $G_i$.
    If $C_1$ and $C_2$ are a linear  $\ell$-intersection pair, then  $\rank(H_1G_2^t)$ and $\rank(G_1H_2^t)$ are independent of $H_i$ and $G_i$ so that
    \[
    \rank(G_1H_2^t)= \rank(H_2G_1^t) =k_1- \ell,
    \]
    and
    \[
    \rank(G_2H_1^t)= \rank(H_1G_2^t)=k_2- \ell.
    \]
\end{theorem}
\begin{proof}
    Assume that $C_1$ and $C_2$ are a linear  $\ell$-intersection pair.
    First, we prove that $ \rank(G_1H_2^t)= \rank(H_2G_1^t) =k_1- \ell$.
    Since $(G_1H_2^t)^t= H_2G_1^t$, it suffices to show that $\rank(G_1H_2^t) =k_1- \ell$.
    
    Since $\ell=\dim(C_1\cap C_2)$,  we have  $n\geq \dim(C_1+C_2) = k_1+k_2-\ell$ which
    implies that  $n- k_2\geq k_1-\ell$ and  $n- k_1\geq k_2-\ell$ .
    Let  $B=\{g_1, g_2,\dots, g_\ell\}$ be a basis of $C_1\cap C_2$.
    If $\ell=k_1 $, then $B\subseteq C_2$ and   $G_1H_2^t=[0]$, and hence $\rank(G_1H_2^t)=0=k_1-\ell$ as desired.
    Assume that  $\ell<k_1 $ and extend  $B$  to be a basis
    $\{g_1,g_2,$ $\dots,g_\ell,$ $g_{\ell+1},\dots, g_{k_1}\}$ for $C_1$.
    Then
    \[
    J_1=\left[
    \begin{array}{c} g_1\\g_2\\
    \vdots \\
    g_\ell \\
    g_{\ell+1}\\
    \vdots\\
    g_{k_1}
    \end{array}
    \right]
    \]
    is a generator  matrix for  $C_1$.
    Applying a suitable sequence of elementary row operations gives an invertible $k_1\times k_1$ matrix $A$ over
    $\mathbb{F}_{q}$ such  that $G_1=AJ_1$ and hence
    \[
    G_1H_2^t= AJ_1H_2^t.
    \]
    Since $A$ is invertible, we have
    \begin{align} \label{eq-rank1}
    \rank(G_1H_2^t)&=\rank(J_1H_2^t).
    \end{align}
    As $g_i\in C_2 $  for all  $i=1,2,\dots, \ell$,  we have  $g_iH_2^t=0 $  for all  $i=1,2,\dots, \ell$ so then
    \begin{align*}
    J_1H_2^t=  \left[
    \begin{array}{c}
    0\\
    \hline
    \left[
    \begin{array}{c}
    g_{\ell+1}\\
    \vdots\\
    g_{k_1}
    \end{array}
    \right] H_2^t
    \end{array}
    \right].
    \end{align*}
    The matrix
    $\left[
    \begin{array}{c}
    g_{\ell+1}\\
    \vdots\\
    g_{k_1}
    \end{array}
    \right] H_2^t$
    has dimensions $(k_1-\ell)\times (n-k_2)$  with   $ n-k_2 \geq k_1-\ell$ so it follows that
    \[\rank\left(\left[
    \begin{array}{c}
    g_{\ell+1}\\
    \vdots\\
    g_{k_1}
    \end{array}
    \right] H_2^t\right) \leq  (k_1-\ell).
    \]
    Suppose that  $\rank\left(\left[
    \begin{array}{c}
    g_{\ell+1}\\
    \vdots\\
    g_{k_1}
    \end{array}
    \right] H_2^t\right) <   k_1-\ell$.
    Then there exists a non-zero vector  $\boldsymbol{u}\in \mathbb{F}_q^{k_1-\ell}$ such that
    \begin{align*} \label{eq-inner0} \boldsymbol{u}\left[
    \begin{array}{c}
    g_{\ell+1}\\
    \vdots\\
    g_{k_1}
    \end{array}
    \right] H_2^t   =[0],
    \end{align*}
    so then   $\boldsymbol{u} \left[
    \begin{array}{c}
    g_{\ell+1}\\
    \vdots\\
    g_{k_1}
    \end{array}
    \right]  \in C_2 \setminus \{\boldsymbol{0}\}$.
    Since  ${\rm span}\{ g_{\ell+1} ,g_{\ell+2},\dots,g_{k_1} \} \cap   C_2 =\{\boldsymbol{0}\} $,  we have $\boldsymbol{u} \left[
    \begin{array}{c}
    g_{\ell+1}\\
    \vdots\\
    g_{k_1}
    \end{array}
    \right] \notin C_2$,
    which is a contradiction.
    Therefore,
    $\rank(G_1H_2^t)=\rank(J_1H_2^t)=k_1-\ell$ which is independent of $G_1$ and $H_2$ as required.
    By swapping $C_1$ and $C_2$, it can be deduced that
    $\rank(G_2H_1^t)= \rank(H_1G_2^t)=k_2- \ell$.
\end{proof}

In the case where the sum of the two codes covers the entire space $\mathbb{F}_q^n$, we have the following corollary.
\begin{corollary}
    For $i\in\{1,2\}$, let $C_i$ be a linear $[n,k_i]_{q}$ code with parity check matrix $H_i$ and generator matrix $G_i$.
    If $C_1$ and $C_2$ are a  linear $\ell$-intersection pair such that $C_1+C_2 =\mathbb{F}_q^n$, then
    \[
    \rank(G_1H_2^t)= \rank(H_2G_1^t) =n-k_2=k_1- \ell,
    \]
    and
    \[
    \rank(G_2H_1^t)= \rank(H_1G_2^t)=n-k_1=k_2- \ell.
    \]
\end{corollary}
\begin{proof}
    Since $C_1+C_2 =\mathbb{F}_q^n$, we have that $n= k_1+k_2- \ell$.
    Then $n-k_2=k_1- \ell$ and $n-k_1=k_2- \ell$, and the results follow from Theorem \ref{thmGLCP}.
\end{proof}

By setting $\ell=0$ in the above corollary, we have the following characterization of LCPs of codes.
\begin{corollary}
    For $i\in\{1,2\}$, let $C_i$ be a linear $[n,k_i]_{q}$ code with parity check matrix $H_i$ and generator matrix $G_i$.
    If $C_1$ and $C_2$ are  a LCP, then
    \[
    \rank(G_1H_2^t)= \rank(H_2G_1^t)  =k_1,
    \]
    and
    \[
    \rank(G_2H_1^t)= \rank(H_1G_2^t)=k_2.
    \]
\end{corollary}

In the case where $C_2$ is the dual code of $C_1$, we have $C_1\cap C_2=Hull(C_1)=Hull(C_2)$ and  the following
result in \cite{GJG} can be obtained from  Theorem \ref{thmGLCP}.
\begin{corollary}[{\cite[Proposition 3.1]{GJG}}]
    \label{cor-hull}
    Let $C$ be an $[n,k]_q$ code over $\mathbb{F}_q$ with generator matrix $G$ and parity-check matrix $H$.
    Then
    \[
    \rank(HH^t)=n-k-\dim(Hull(C)),
    \]
    and
    \[
    \rank(GG^t)=k- \dim(Hull(C)).
    \]
\end{corollary}

The well-known properties of self-orthogonal codes and complementary dual codes \cite{Mac} follow easily from Corollary \ref{cor-hull}.
\begin{corollary}
    Let $C$ be an $[n,k]_q$ code over $\mathbb{F}_q$ with generator matrix $G$ and parity-check matrix $H$.
    Then the following statements hold.
    \begin{enumerate}
        \item $C$ is self-orthogonal if and only if  $GG^t=[0]$.
        \item $C$ is complementary dual  if and only if  $GG^t$ is invertible.
        In this case, $HH^t$ is invertible.
    \end{enumerate}
\end{corollary}

\begin{remark}
    In general, we may relate a linear $\ell$-intersection pair of codes with the Galois dual of a linear code \cite{FZ2017}.
    For $q=p^e$  and $0\leq h<e$, the {\em $p^h$-inner product  (Galois inner product)} between
    $\boldsymbol{u}=(u_1,u_2,\dots, u_n)$ and $\boldsymbol{v}=(v_1,v_2,\dots, v_n)$  in $\mathbb{F}_q$ is defined to be
    \[
    \langle \boldsymbol{u}, \boldsymbol{v} \rangle_h=\sum _{i=1}^n u_iv_i^{p^h}.
    \]
    The {\em $p^h$-dual (Galois dual)}  $C^{\perp_h}$ of a linear code $C$ is defined as
    \[
    C^{\perp_h}=\{  \boldsymbol{u} \in \mathbb{F}_q^n \mid \langle \boldsymbol{u}, \boldsymbol{c} \rangle_h =0 \text{ for all } \in C\}.
    \]
    Note that $C^{\perp_0}$ is the Euclidean dual $C^\perp$.
    If $e$ is even,  the $C^{\perp_{\frac{e}{2}}}$ is the well-known Hermitian dual.
    
    Using statements similar to those in the proof of  Theorem  \ref{thmGLCP},  the following  result can be concluded.
    For $i\in\{1,2\}$, let $C_i$ be a linear $[n,k_i]_{q}$ code with generator matrix $G_i$ and let $H_i$ be a generator matrix for the Galois dual $C_i^{\perp_h}$.
    If $C_1$ and $C_2$ are a linear $\ell$-intersection pair then
    \[
    \rank(G_1H_2^*)= \rank(H_2G_1^*) =k_1- \ell,
    \]
    and
    \[
    \rank(G_2H_1^*)= \rank(H_1G_2^*)=k_2- \ell,
    \]
    where $A^*=[a_{ji}^{p^h}]$ for a matrix $A=[a_{ij}]$ over $\mathbb{F}_q$.
\end{remark}

\subsection{Constructions of Linear $\ell$-Intersection Pairs of Codes}

In this subsection, a discussion on constructions of  linear $\ell$-intersection pairs is given. 
We note that  constructions of     linear $0$-intersection pairs of  linear codes $C_1$ and $C_2$  with $\dim(C_1)+\dim(C_2)=n$,   LCPs of codes,    have been given in  \cite{carlet2}. Various  constructions of 
linear $0$-intersection pairs  of  linear codes  $C_1$ and $C_2=C_1^\perp$,    LCD codes,  have  been  discussed in  \cite{CG2016}, \cite{CMTQ2018}, \cite{Jin2017},  \cite{M1992} and  \cite{QZ2018}. Constructions of some  linear codes with prescribed hull dimension have been given in  \cite{GJG} and
\cite{LC2018}.

Based on the characterizations given in Subsection \ref{ssec2.1}, the value $\ell$ for which    two linear codes of   length $n$  over  $\mathbb{F}_q$  form a linear $\ell$-intersection pair can be easily determined. Here, constructions of  linear $\ell$-intersection pairs will be given using the concept of equivalent codes and some propagation rules.

We recall that 
two  linear  codes  of length $n$ over $\mathbb{F}_q$ are {\em equivalent} if one can be obtained from the other by a combination of operations of the following types: (i) permutation of the $n$ digits of the codewords; (ii) multiplication of the symbols appearing in a fixed position by a nonzero  element in $\mathbb{F}_q$.  A square  matrix over $\mathbb{F}_q$  is called a  {\it weighted  permutation matrix}  if it  has exactly one nonzero entry  in each row and each column and $0$s elsewhere. It is not difficult to see that  linear codes $C_1$ and $C_2$ of length $n$ over $\mathbb{F}_q$ are  equivalent if and only if there exists an $n\times n$ weighted  permutation matrix $A$  over $\mathbb{F}_q$ such that $C_2=\{\boldsymbol{c} A\mid \boldsymbol{c} \in C_1\}$.   This characterization  is useful in constructions of linear $\ell$-intersection pairs.

\begin{lemma}\label{rem1} Let $C_1$  and $C_2$ be $[n,k_1,d_1]_q$ and $[n,k_2,d_2]_q$  codes, respectively. Let $A$ be an $n\times n$ weighted permutation matrix over $\mathbb{F}_q$ and let $G_1$ and $H_2$ be a generator matrix of $C_1$ and a parity-check matrix of $C_2$, respectively. Then their exists a linear $\ell$-intersection pair of $[n,k_1,d_1]_q$ and $[n,k_2,d_2]_q$  codes, where $\ell=k_1-\rank(G_1AH_2^t)$. 
\end{lemma}
\begin{proof} Let $C_1'$ be  the linear code  generated by $G_1A$. By the discussion above, $C'_1$ is equivalent to $C_1$. Hence, $C_1'$ is an $[n,k_1,d_1]_q$ code.   By  Theorem \ref{thmGLCP},  $C_1'$ and $C_2$ form a  linear $\ell$-intersection pair of $[n,k_1,d_1]_q$ and $[n,k_2,d_2]_q$  codes, where $\ell=k_1-\rank((G_1A)H_2^t) =k_1-\rank(G_1AH_2^t)$. 
\end{proof}

In Lemma \ref{rem1}, the value $\ell$ depends on the choices of $A$.   In applications,    a suitable weighted permutation matrix $A$ is required. Illustrative examples are given as follows.

\begin{example} \label{exwp} Let   $C_1$ and $C_2$ be $[7,4,3]_2$ and $[7,3,4]_2$ codes with generator matrices 
    \[ G_1= \begin{bmatrix}
    1 &0 &0 &0 &0 &1& 1\\
    0 &1 &0 &0 &1 &0 &1\\
    0 &0 &1& 0 &1& 1 &0\\
    0 &0& 0& 1& 1 &1& 1 \end{bmatrix}
    \text{ and } G_2=\begin{bmatrix}
    1&0 &1& 0&1&0& 1\\
    0 &1 &1 &0 &0& 1& 1\\
    0& 0 &0 &1 &1& 1 &1
    \end{bmatrix}.
    \]     
    Using the computer algebra system MAGMA  \cite{BCP1997} and Theorem \ref{thmGLCP}, it can be seen  that  $C_1$ and $C_2$ form a linear $3$-intersection pair of  $[7,4,3]_2$ and $[7,3,4]_2$ codes.

    Let \[ A_1=\begin{bmatrix}
    1&0&0&0&0&0&0\\
    0&1&0&0&0&0&0\\
    0&0&1&0&0&0&0\\
    0&0&0&1&0&0&0\\
    0&0&0&0&1&0&0\\
    0&0&0&0&0&0&1\\
    0&0&0&0&0&1&0
    \end{bmatrix},~~
    \begin{bmatrix}
    0&1&0&0&0&0&0\\
    1&0&0&0&0&0&0\\
    0&0&1&0&0&0&0\\
    0&0&0&1&0&0&0\\
    0&0&0&0&1&0&0\\
    0&0&0&0&0&0&1\\
    0&0&0&0&0&1&0
    \end{bmatrix}
    \text{ and }
    \begin{bmatrix}
    0&0&0&0&0&0&1\\
    1&0&0&0&0&0&0\\
    0&0&1&0&0&0&0\\
    0&0&0&1&0&0&0\\
    0&0&0&0&1&0&0\\
    0&1&0&0&0&0&0\\
    0&0&0&0&0&1&0
    \end{bmatrix}
    \] be $7\times 7$ (weighted) permutation matrices over $\mathbb{F}_2$. 
    Let $C_1'$, $C_1''$  and $C_1'''$ be linear codes generated by $G_1A_1$, $G_1A_2$  and $G_1A_3$, respectively. 
    Using the computer algebra system MAGMA  \cite{BCP1997} and  Lemma \ref{rem1},  we have that $C_1'$ and $C_2$ form a linear $2$-intersection pair of  $[7,4,3]_2$ and $[7,3,4]_2$ codes,  $C_1''$ and $C_2$ form a linear $1$-intersection pair of  $[7,4,3]_2$ and $[7,3,4]_2$ codes, and $C_1'''$ and $C_2$ form a linear $0$-intersection pair of  $[7,4,3]_2$ and $[7,3,4]_2$ codes.
\end{example}

Next,     useful   recursive  constructions of linear $\ell$-intersection pairs are  given.

\begin{theorem} \label{thm2.2} Let  $\ell\geq 0$ be an integer. If there exists a linear $\ell$-intersection pair of  $[n,k_1,d_1]_{q}$ and    $[n,k_2,d_2]_{q}$ codes,  then the following statements hold.
    \begin{enumerate}
        \item   There exists a linear $\gamma$-intersection pair of  $[n,k_1,d_1]_{q}$ and    $[n,k_2-\ell+\gamma ,D_2]_{q}$ codes for all $0\leq \gamma \leq \ell$, where   $D_2\geq d_2$.
        \item   There exists a linear $\gamma$-intersection pair of  $[n+\ell-\gamma,k_1,d_1]_{q}$ and    $[n+\ell-\gamma,k_2 ,D_2]_{q}$ codes for all $0\leq \gamma \leq \ell$, where   $D_2\geq d_2$.
    \end{enumerate}
    
\end{theorem}
\begin{proof} Assume that there exists a linear $\ell$-intersection pair of  $[n,k_1,d_1]_{q}$ and    $[n,k_2,d_2]_{q}$ codes, denoted by $C_1$ and $C_2$, respectively. Let $A=\{\boldsymbol{v}_1, \boldsymbol{v}_2,\dots, \boldsymbol{v}_{\ell}\}$ be a basis of  $C_1\cap C_2$. Let $B_1$ and $B_2$ be bases of $C_1$ and $C_2$ extended respectively from $A$.  For $\gamma=\ell$, the two statements are obvious.  Assume that  $0\leq \gamma <\ell$. 
    
    To prove 1, let $C_2'$ be the  linear code  generated by $B_2\setminus \{\boldsymbol{v}_1, \boldsymbol{v}_2,\dots, \boldsymbol{v}_{\ell-\gamma}\}$.  Then   $C_2'$  is an  $[n, k_2-\ell+\gamma]_q$ code.  Since $C_2'$ is a subcode of $C_2$,  we have $d(C_2')=D_2$  for some $D_2\geq d_2$. It is clear that $C_1$ and $C_2'$ form  a linear $\gamma$-intersection pair.

    To prove 2,   let  $\varphi_1: B_1 \to \mathbb{F}_q^{n+1}$  and $\varphi_2:  B_2 \to \mathbb{F}_q^{n+1}$ be concatenated maps  defined by 
    \[ \varphi_1( \boldsymbol{u})= 
    \boldsymbol{u}|0  \]  for all $ \boldsymbol{u}\in B_1$, and 
    
    \[ \varphi_2( \boldsymbol{u})=\begin{cases}
    \boldsymbol{u}|1 & \text{ if } \boldsymbol{u}=\boldsymbol{v}_{\ell},\\
    \boldsymbol{u}|0 & \text{  otherwise}
    \end{cases}\]
    for all $ \boldsymbol{u}\in B_2$.   Let $C_1' $ and $C_2'$ be the linear codes generated by $\varphi_1(B_1)$ and  $\varphi_2(B_2)$. Clearly, $C_1'$ and $C_2'$ form a linear $(\ell-1)$-intersection pair of  $[n+1,k_1,d_1]_{q}$ and    $[n+1,k_2 ,D_2]_{q}$ codes for some  $D_2\geq d_2$.  Continue this process, a linear $\gamma$-intersection pair of  $[n+\ell-\gamma,k_1,d_1]_{q}$ and    $[n+\ell-\gamma,k_2 ,D_2]_{q}$ codes can be constructed for all $0\leq \gamma < \ell$, where   $D_2\geq d_2$.
\end{proof}

Based on the characterizations given in Subsection \ref{ssec2.1}, Lemma \ref{rem1}  and the best  known linear codes in  \cite{G2019}, some linear $\ell$-intersection pair of  good codes over small finite fields can be constructed using the following steps:

\begin{enumerate}[~~~~1)]
    \item Fix two best known linear   codes  $C_1$ and $C_2$ of length $n$ over $\mathbb{F}_q$ from \cite{G2019}.
    \item Fix an $n\times n $ weighted permutation matrix $A$ over $\mathbb{F}_q$.
    \item Compute $C_1'=\{\boldsymbol{c} A\mid \boldsymbol{c} \in C_1\}$.
    \item Compute  the value $\ell$ for which  $C_1'$ and $C_2$  form a linear $\ell$-intersection pair using Lemma \ref{rem1}.   \\ Output: linear $\ell$-intersection pair.
    \item Apply recursive constructions given in Theorem \ref{thm2.2}.  \\
    Output:  linear $\gamma$-intersection pair, where $0\leq \gamma \leq \ell$.
\end{enumerate}

We note that  a linear $\ell$-intersection pair of linear codes with best known parameters is obtained in Step 4   while the minimum distance of the second code in   a linear $\gamma$-intersection pair  obtained in  Step 5  might be lower than the best known ones.

Using basic  linear algebra,  we have  the following result.
\begin{lemma} \label{rage-l} If  there exists a linear $\ell$-intersection pair of $[n,k_1,d_1]_q$ and  $[n,k_2,d_2]_q$ codes, then $k_1+k_2-n\leq \ell \leq \min\{k_1,k_2\}$. 
\end{lemma}
Note that  Lemma \ref{rage-l}  does not guarantee the existence of   a linear $\ell$-intersection pair of $[n,k_1,d_1]_q$ and  $[n,k_2,d_2]_q$ codes for all $\ell$ satisfying  $k_1+k_2-n\leq \ell \leq \min\{k_1,k_2\}$. For the existence of such pairs,   we propose the following conjecture.

\begin{conjecture} \label{cojecture}
    There exists a linear $\ell$-intersection pair of $[n,k_1,d_1]_q$ and  $[n,k_2,d_2]_q$ codes for all $\ell$ satisfying  $k_1+k_2-n\leq \ell \leq \min\{k_1,k_2\}$ provided that there exist $[n,k_1,d_1]_q$ and  $[n,k_2,d_2]_q$ codes.
\end{conjecture}
This conjecture holds true for MDS codes over $\mathbb{F}_q$ of length less than or equal to $q+1 $ and it will be proved in Section \ref{sectionMDS}. The other cases remain   an open problem. 
In our view, the concept of equivalent codes in Lemma \ref{rem1} might be useful in solving Conjecture~\ref{cojecture} as discussed  in Example \ref{exwp}.

%
%
%
%
%
%

\section{Linear $\ell$-Intersection Pairs of  MDS Codes} \label{sectionMDS}
The class of MDS codes has been extensively studied for decades, and numerous approach have been developed to construct these codes.
Such codes are optimal in the sense that for a fixed length and dimension they can correct
and detect the maximum  number of errors.
Further, they have various applications in erasure channel and wireless communication.
In this section, constructions of linear  $\ell$-intersection pairs of MDS codes for all possible parameters over $\mathbb{F}_q$  are given up to length $q+1$. Precisely, Conjecture \ref{cojecture} is proved for MDS codes of    such lengths.

\subsection{Linear $\ell$-Intersection Pairs of  MDS Codes from GRS and Extended GRS Codes}
We focus here on constructions of linear $\ell$- intersection pairs of MDS codes.
Based on the notation given  in \cite{Jin2017}, all linear  $\ell$-intersection pairs of MDS codes   up to length $q+1$ over $\mathbb{F}_q$
can be derived in terms of generalized Reed-Solomon (GRS) codes and extended GRS codes.

First, we determine the intersection of two arbitrary GRS codes of length up to $q$.
For $n\leq q$, let
$\boldsymbol{a}=(a_1,a_2,\dots,a_n)$ and $\boldsymbol{v}=(v_1,v_2,\dots,v_n)$
be words in $\mathbb{F}_q^n$ such that $a_1,a_2,\dots,a_n$ are distinct and  $v_1,v_2,\dots,v_n$ are non-zero.
Let $P(x)$ be a non-zero polynomial in $\mathbb{F}_q[x]$ of degree less than or equal to $n$.
If $P(a_i)\ne 0$ for all $i=1,2,\dots,n$, the code (see \cite{Jin2017})
\begin{align}
{\rm  GRS}(\boldsymbol{a},P(x), \boldsymbol{v})=&\left\{ \left(  \frac{v_1f(a_1)}{P(a_1)},   \frac{v_2f(a_2)}{P(a_2)}, \dots, \frac{v_nf(a_n)}{P(a_n)} \right) \, \Big\vert \,  f(x)\in \mathbb{F}_q[x] \right. \notag  \\
&\hskip10em \left. \text{and } \deg(f(x))< \deg(P(x))  \vphantom{\frac{v_1f(a_1)}{P(a_1)}}  \right\}
\end{align}
is well defined and is a  GRS code  of length $n$ and dimension $\deg(P(x))$.
From \cite{Jin2017}, this code is  MDS with parameters $[n,\deg(P(x)), n- \deg(P(x))+1 ]_q$.

For polynomials $P(x)$ and $Q(x)$ over $\mathbb{F}_q$ of degree less than or equal to $n$, some properties of the codes
${\rm  GRS}(\boldsymbol{a},P(x), \boldsymbol{v})$ and ${\rm  GRS}(\boldsymbol{a},Q(x), \boldsymbol{v})$
can be determined in terms of $P(x)$ and $Q(x)$.

\begin{lemma}
    \label{PdivQ}
    Let $q$ be a prime power and $n\leq q$ be an integer.
    Further, let $P(x)$ and $Q(x)$ be polynomials of degree less than or equal to $n$  in $\mathbb{F}_q[x]$
    such that  $\gcd(P(x)Q(x), \prod\limits_{i=1}^n(x-a_i))=1$.
    If $P(x)|Q(x)$, then \[  {\rm  GRS}(\boldsymbol{a},P(x), \boldsymbol{v}) \subseteq  {\rm  GRS}(\boldsymbol{a},Q(x), \boldsymbol{v}).\]
\end{lemma}

\begin{proof}  Note that  the GRS codes  ${\rm  GRS}(\boldsymbol{a},P(x), \boldsymbol{v}) $
    and $ {\rm  GRS}(\boldsymbol{a},Q(x), \boldsymbol{v})$ are well defined since $\gcd(P(x)Q(x), \prod\limits_{i=1}^n(x-a_i))=1$.
    Assume that $P(x)|Q(x))$ and let
    \[
    \boldsymbol{c}= \left(  \frac{v_1f(a_1)}{P(a_1)},   \frac{v_2f(a_2)}{ P(a_2)}, \dots, \frac{v_nf(a_n)}{P(a_n)} \right) \in   {\rm  GRS}(\boldsymbol{a},P(x), \boldsymbol{v}).
    \]
    Then $\deg(f(x))<\deg(P(x))$  which implies that $\frac{f(x)Q(x) }{ P(x)}$ is a polynomial over $\mathbb{F}_q$ and  $\deg(\frac{f(x)Q(x) }{ P(x)})<\deg(Q(x))$.  It follows that  \begin{align*}\boldsymbol{c}&= \left(  \frac{v_1f(a_1)}{P(a_1)},   \frac{v_2f(a_2)}{P(a_2)}, \dots, \frac{v_nf(a_n)}{P(a_n)} \right)\\
    & = \left(  \frac{v_1\frac{f(a_1)Q(a_1)}{P(a_1)}}{Q(a_1)},   \frac{v_2 \frac{f(a_2)Q(a_2)}{P(a_2)}}{Q(a_2)}, \dots, \frac{v_n \frac{f(a_n)Q(a_n)}{P(a_n)}}{Q(a_n)} \right)
    \in   {\rm  GRS}(\boldsymbol{a},P(x), \boldsymbol{v}),
    \end{align*}
    and hence $ {\rm  GRS}(\boldsymbol{a},P(x), \boldsymbol{v}) \subseteq  {\rm  GRS}(\boldsymbol{a},Q(x), \boldsymbol{v})$ as required.
\end{proof}

The intersection of pairs of GRS codes is determined in the next theorem.

\begin{theorem}
    \label{PQ-intersection}
    Let $q$ be a prime power and $n\leq q$ be a positive  integer.
    Further, let $P(x)$, $Q(x)$ and $L(x)$ be polynomials of degree less than or equal to $n$  in $\mathbb{F}_q[x]$ satisfying  the following conditions.
    \begin{enumerate}
        \item  $\gcd(P(x),Q(x))=L(x)$.
        \item $\gcd(P(x)Q(x), \prod\limits_{i=1}^n(x-a_i))=1$.
        \item $\deg(P(x))+\deg(Q(x)\leq n+\deg(L(x))$.
    \end{enumerate}
    Then $  {\rm  GRS}(\boldsymbol{a},P(x), \boldsymbol{v})\cap  {\rm  GRS}(\boldsymbol{a},Q(x), \boldsymbol{v}) =   {\rm  GRS}(\boldsymbol{a},L(x), \boldsymbol{v})$.
\end{theorem}
\begin{proof} Since $\gcd(P(x)Q(x), \prod\limits_{i=1}^n(x-a_i))=1$,   the three codes   are  well defined.  Note that $L(x)|P(x)$, $L(x)|Q(x)$ and   $\gcd(\frac{P(x)}{L(x)},\frac{Q(x)}{L(x)})=1$.
    Let   \begin{align*}
    \boldsymbol{c}&= \left(  \frac{v_1f(a_1)}{P(a_1)},   \frac{v_2f(a_2)}{P(a_2)}, \dots, \frac{v_nf(a_n)}{P(a_n)} \right)  \\
    &=  \left(  \frac{v_1g(a_1)}{Q(a_1)},   \frac{v_2g(a_2)}{Q(a_2)}, \dots, \frac{v_ng(a_n)}{Q(a_n)} \right)  \in  {\rm  GRS}(\boldsymbol{a},P(x), \boldsymbol{v})\cap  {\rm  GRS}(\boldsymbol{a},Q(x), \boldsymbol{v})
    \end{align*}
    for some $f(x)$ and $g(x)$ in $\mathbb{F}_q[x]$ such that $\deg(f(x))< \deg(P(x))$  and  $\deg(f(x))< \deg(Q(x))$.
    Then $ {f(a_i)}{Q(a_i)}= {g(a_i)}{P(a_i)}$ for all $i=1,2,\dots,n$  which implies that  $ {f(a_i)}\frac{Q(a_i)}{L(a_i)}- {g(a_i)}\frac{P(a_i)}{L(a_i)}=0$
    for all $i=1,2,\dots,n$.
    Since $\deg(f(x))< \deg(P(x))$  and  $\deg(f(x))< \deg(Q(x))$, we have $\deg(f(x) \frac{Q(x)}{L(x)})< \deg(P(x))+\deg(Q(x))-\deg(L(x)) \leq n$  and   $\deg(g(x) \frac{P(x)}{L(x)})< \deg(Q(x))+\deg(P(x))-\deg(L(x)) \leq n$.  It follows that   $ {f(x)}\frac{Q(x)}{L(x)}- {g(x)}\frac{P(x)}{L(x)}=0$ is the zero polynomial.  Since $\gcd(\frac{P(x)}{L(x)},\frac{Q(x)}{L(x)})=1$, we have    $\frac{Q(x)}{L(x)}|g(x)  $ and $ \frac{P(x)}{L(x)}|f(x)$.    Then $\frac{f(x)L(x)}{P(x)}$ is a polynomial in $\mathbb{F}_q[x]$ such that   $\deg (\frac{f(x)L(x)}{P(x)})< \deg(L(x))$ and
    \[ \frac{v_if(a_i)}{P(a_i)}=  \frac{v_if(a_i)L(a_i)/ P(a_i)}{L(a_i)} \]
    for all $i=1,2,\dots,n$.
    Hence,
    \begin{align*}
    \boldsymbol{c}&= \left(   \frac{v_1f(a_1)L(a_1)/ P(a_1)}{L(a_1)},     \frac{v_2f(a_2)L(a_2)/ P(a_2)}{L(a_2)} , \dots,  \frac{v_nf(a_n)L(a_n)/ P(a_n)}{L(a_n)} \right)\\
    &  \in {\rm  GRS}(\boldsymbol{a},L(x), \boldsymbol{v}),
    \end{align*}
    and therefore,
    ${\rm  GRS}(\boldsymbol{a},P(x), \boldsymbol{v})\cap  {\rm  GRS}(\boldsymbol{a},Q(x), \boldsymbol{v})  \subseteq  {\rm  GRS}(\boldsymbol{a},L(x), \boldsymbol{v}) $.
    
    Since  $L(x)|P(x)$ and $L(x)|Q(x)$,  we have   $  {\rm  GRS}(\boldsymbol{a},L(x), \boldsymbol{v}) \subseteq   {\rm  GRS}(\boldsymbol{a},P(x), \boldsymbol{v})$ and  $  {\rm  GRS}(\boldsymbol{a},L(x), \boldsymbol{v}) \subseteq   {\rm  GRS}(\boldsymbol{a},Q(x), \boldsymbol{v})$  by Lemma \ref{PdivQ}.
    Hence, $  {\rm  GRS}(\boldsymbol{a},L(x), \boldsymbol{v}) \subseteq    {\rm  GRS}(\boldsymbol{a},P(x), \boldsymbol{v})\cap  {\rm  GRS}(\boldsymbol{a},Q(x), \boldsymbol{v}) $.
    Therefore,
    ${\rm  GRS}(\boldsymbol{a},P(x), \boldsymbol{v})\cap  {\rm  GRS}(\boldsymbol{a},Q(x), \boldsymbol{v}) =   {\rm  GRS}(\boldsymbol{a},L(x), \boldsymbol{v})$
    as required.
\end{proof}

The following corollaries can be derived from Theorem \ref{PQ-intersection}.

\begin{corollary}
    \label{lcdPQ}
    Let $q$ be a prime power and $n\leq q$ be a positive integer.
    Further, let $P(x)$, $Q(x)$ and $M(x)$ be polynomials of degree less than or equal to $n$  in $\mathbb{F}_q[x]$ satisfying  the following conditions.
    \begin{enumerate}
        \item  ${\rm lcm}(P(x),Q(x))=M(x)$.
        \item $\gcd(P(x)Q(x), \prod\limits_{i=1}^n(x-a_i))=1$.
        \item $ \deg(M(x))\leq  n$.
    \end{enumerate}
    Then
    $  {\rm  GRS}(\boldsymbol{a},P(x), \boldsymbol{v})+  {\rm  GRS}(\boldsymbol{a},Q(x), \boldsymbol{v}) =   {\rm  GRS}(\boldsymbol{a},M(x), \boldsymbol{v})$.
\end{corollary}
\begin{proof}
    Since $P(x)|M(x)$ and $Q(x)|M(x)$, we have
    $  {\rm  GRS}(\boldsymbol{a},P(x), \boldsymbol{v}) \subseteq    {\rm  GRS}(\boldsymbol{a},M(x), \boldsymbol{v})$ and
    $  {\rm  GRS}(\boldsymbol{a},Q(x), \boldsymbol{v}) \subseteq    {\rm  GRS}(\boldsymbol{a},M(x), \boldsymbol{v})$ by Lemma \ref{PdivQ},
    and hence
    \[
    {\rm  GRS}(\boldsymbol{a},P(x), \boldsymbol{v})+  {\rm  GRS}(\boldsymbol{a},Q(x), \boldsymbol{v})  \subseteq  {\rm  GRS}(\boldsymbol{a},M(x), \boldsymbol{v}).
    \]
    From Theorem \ref{PQ-intersection}, we have  \begin{align*}
    \dim (  {\rm  GRS}(\boldsymbol{a},P(x), \boldsymbol{v})&+  {\rm  GRS}(\boldsymbol{a},Q(x), \boldsymbol{v})  )\\
    &= \dim( {\rm  GRS}(\boldsymbol{a},P(x), \boldsymbol{v})) + \dim( {\rm  GRS}(\boldsymbol{a},Q(x), \boldsymbol{v}) )\\
    &\quad -\dim({\rm  GRS}(\boldsymbol{a},\gcd(P(x),Q(x)), \boldsymbol{v})  )
    \\
    &=\deg(P(x))+\deg(Q(x)) -\deg(\gcd(P(x),Q(x)))\\
    &= \deg(M(x))\\
    &= \dim({\rm  GRS}(\boldsymbol{a},M(x), \boldsymbol{v})),
    \end{align*}
    so ${\rm  GRS}(\boldsymbol{a},P(x), \boldsymbol{v})+  {\rm  GRS}(\boldsymbol{a},Q(x), \boldsymbol{v}) =   {\rm  GRS}(\boldsymbol{a},M(x), \boldsymbol{v})$ as required.
\end{proof}

The existence of linear $\ell$-intersection pairs of GRS (MDS) codes with prescribed parameters depends on the existence of polynomials $P(x)$, $Q(x)$ and $L(x)$ which will be discussed in Theorem \ref{prop-MDS-pair}.

The algebraic structure of the intersection of arbitrary pairs of  GRS codes of length up to $q$ was discussed above.
Next, we focus on the intersection of pairs of MDS codes of length $n\leq q+1$ using extended GRS codes.
For a positive integer  $n\leq q+1$,
let $\boldsymbol{a}=(a_1,a_2,\dots,a_{n-1}) \in \mathbb{F}_q^{n-1}$ and $\boldsymbol{v}=(v_1,v_2,\dots,v_{n-1})\in \mathbb{F}_q^{n}$
be  such that $a_1,a_2,\dots,a_{n-1}$ are distinct and  $v_1,v_2,\dots,v_{n-1}$ are non-zero.
Let $P(x)$ be a non-zero polynomial in $\mathbb{F}_q[x]$ of degree  $k\leq n$ such that $P(a_i)\ne 0$ for all $i=1,2,\dots,{n-1}$.

For polynomials $a(x)=a_0+a_1x+\dots+ a_tx^t$ and $b(x)=b_0+b_1x+\dots+b_tx^t$ in $\mathbb{F}_q[x]$ with $b_t\ne0$,
let  $r(x)=\frac{a(x)}{b(x)}$ be a rational function.
The evaluation $r(\infty) $ is defined to be $\frac{a_t}{b_t}$.
Thus, $r(\infty) =0$ if and only if $\deg(a(x))<\deg(b(x))$.
For a polynomial  $f(x)\in \mathbb{F}_q[x]$ with $\deg(f(x))<\deg(P(x))$,
we then have  $  \left(\frac{xf}{P}\right)(\infty)=0 $ if and only if $\deg(f(x))\leq \deg(P(x))-2$,
and  $\left( \frac{xf}{P}\right)(\infty)\ne 0 $ if and only if $\deg(f(x))= \deg(P(x))-1$.

In \cite{Jin2017}, the extended GRS code  of length $n$ and dimension $\deg(P(x))$ over $\mathbb{F}_q$  is defined to be
\begin{align}
{\rm  GRS}_\infty&(\boldsymbol{a},P(x), \boldsymbol{v})
\notag\\&=\left\{ \left(  \frac{v_1f(a_1)}{P(a_1)},   \frac{v_2f(a_2)}{P(a_2)}, \dots, \frac{v_{n-1}f(a_{n-1})}{P(a_{n-1})},  v_{n} \left( \frac{xf}{P}\right)(\infty) \right) \, \Big\vert \, f(x)\in \mathbb{F}_q[x]  \right. \notag \\
& \hskip10em \left.\text{and } \deg(f(x))< \deg(P(x))  \vphantom{\frac{v_1f(a_1)}{P(a_1)}}  \right\}.
\end{align}
From \cite{Jin2017}, this extended GRS code is  MDS with parameters $[n,\deg(P(x)), n- \deg(P(x))+1 ]_q$.

The properties of extended GRS codes are now given in terms of the corresponding polynomial $P(x)$.
\begin{lemma}
    \label{EPdivQ}
    Let $q$ be a prime power and $n\leq q+1$ be a positive integer.
    Further, let $P(x)$ and $Q(x)$ be polynomials of degree less than or equal to $n$  in $\mathbb{F}_q[x]$
    such that  $\gcd(P(x)Q(x), \prod\limits_{i=1}^n(x-a_i))=1$.
    If $P(x)|Q(x)$, then \[  {\rm  GRS}_\infty(\boldsymbol{a},P(x), \boldsymbol{v}) \subseteq  {\rm  GRS}_\infty(\boldsymbol{a},Q(x), \boldsymbol{v}).\]
\end{lemma}
\begin{proof}
    Assume that $P(x)|Q(x))$.
    Let \[\boldsymbol{c}= \left(  \frac{v_1f(a_1)}{P(a_1)},   \frac{v_2f(a_2)}{ P(a_2)}, \dots, \frac{v_{n-1}f(a_{n-1})}{P(a_{n-1})},  v_{n} \left( \frac{xf}{P}\right)(\infty) \right) \in   {\rm  GRS}_\infty(\boldsymbol{a},P(x), \boldsymbol{v}).\]
    Using arguments similar to those in the proof of Lemma \ref{PdivQ},  we have that $\frac{f(x)Q(x) }{ P(x)}$ is a polynomial in $\mathbb{F}_q$ and  $\deg(\frac{f(x)Q(x) }{ P(x)})<\deg(Q(x))$. Moreover,  $\frac{v_if(a_i)}{P(a_i)}= \frac{v_i\frac{f(a_i)Q(a_i)}{P(a_i)}}{Q(a_i)}$ for all $i=1,2,\dots, n-1$.     Since $\deg( \frac{xf(x)Q(x)}{P(x)}) -\deg(Q(x))= \deg(xf(x))-\deg(P(x))$, we have
    $\left( \frac{xf}{P}\right)(\infty)=  \left( \frac{\frac{xfQ}{P}}{Q}\right)(\infty)$.
    It follows that
    $\boldsymbol{c}
    \in   {\rm  GRS}_\infty(\boldsymbol{a},P(x), \boldsymbol{v})$,
    and hence ${\rm  GRS}_\infty (\boldsymbol{a},P(x), \boldsymbol{v}) \subseteq  {\rm  GRS}_\infty (\boldsymbol{a},Q(x), \boldsymbol{v})$ as required.
\end{proof}

\begin{theorem}
    \label{EPQ-intersection}
    Let $q$ be a prime power and $n\leq q+1$ be a positive integer.
    Further, let $P(x)$, $Q(x)$ and $L(x)$ be polynomials of degree less than or equal to $n$  in $\mathbb{F}_q[x]$ satisfying  the following conditions.
    \begin{enumerate}
        \item  $\gcd(P(x),Q(x))=L(x)$.
        \item $\gcd(P(x)Q(x), \prod\limits_{i=1}^n(x-a_i))=1$.
        \item $\deg(P(x))+\deg(Q(x))\leq n+\deg(L(x))$.
    \end{enumerate}
    Then ${\rm  GRS}_\infty(\boldsymbol{a},P(x), \boldsymbol{v})\cap  {\rm  GRS}_\infty(\boldsymbol{a},Q(x), \boldsymbol{v}) =   {\rm  GRS}_\infty(\boldsymbol{a},L(x), \boldsymbol{v})$.
\end{theorem}
\begin{proof} First, we note that   $L(x)|P(x)$, $L(x)|Q(x)$ and   $\gcd(\frac{P(x)}{L(x)},\frac{Q(x)}{L(x)})=1$.    Since $\gcd(P(x)Q(x), \prod\limits_{i=1}^n(x-a_i))=1$,  we have $P(a_i)\ne 0$, $Q(a_i)\ne 0$, $L(a_i)\ne 0$ and the three codes   are  well defined.
    Let
    \begin{align*}
    \boldsymbol{c}&= \left(  \frac{v_1f(a_1)}{P(a_1)},   \frac{v_2f(a_2)}{P(a_2)}, \dots, \frac{v_{n-1}f(a_{n-1})}{P(a_{n-1})}, v_{n} \left( \frac{xf}{P}\right)(\infty)  \right)  \\
    &=  \left(  \frac{v_1g(a_1)}{Q(a_1)},   \frac{v_2g(a_2)}{Q(a_2)}, \dots, \frac{v_{n-1}g(a_{n-1})}{Q(a_{n-1})} , v_{n} \left( \frac{xg}{Q}\right)(\infty) \right)  \\
    &\in  {\rm  GRS}_\infty(\boldsymbol{a},P(x), \boldsymbol{v})\cap  {\rm  GRS}_\infty(\boldsymbol{a},Q(x), \boldsymbol{v})
    \end{align*}
    for some  $f(x)$ and $g(x)$ in $\mathbb{F}_q[x]$ such that $\deg(f(x))< \deg(P(x))$  and  $\deg(f(x))< \deg(Q(x))$.
    Then
    \begin{align}
    \label{eqn-1} {f(a_i)}{Q(a_i)}= {g(a_i)}{P(a_i)},
    \end{align}
    for all $i=1,2,\dots,n-1$   and  $ \left( \frac{xf}{P}\right)(\infty)= \left( \frac{xg}{Q}\right)(\infty)$.
    
    The latter condition $ \left( \frac{xf}{P}\right)(\infty)= \left( \frac{xg}{Q}\right)(\infty)$ implies that
    $\deg( f(x)Q(x)-g(x)P(x))\leq \deg(P(x)Q(x) )-2=  n+\deg(L(x))-2$.
    Hence, we have  $\deg( f(x)\frac{Q(x)}{L(x)}-g(x)\frac{P(x)}{L(x)})\leq   n -2$  so that  $ \left( \frac{xf}{P}\right)(\infty)= \left( \frac{xfL/P}{L}\right)(\infty)$.
    From \eqref{eqn-1},  we have   $ {f(a_i)}\frac{Q(a_i)}{L(a_i)}- {g(a_i)}\frac{P(a_i)}{L(a_i)}=0$
    for all $i=1,2,\dots,n-1$.
    It follows that   $ {f(x)}\frac{Q(x)}{L(x)}- {g(x)}\frac{P(x)}{L(x)}=0$ is the zero polynomial.
    Since $\gcd(\frac{P(x)}{L(x)},\frac{Q(x)}{L(x)})=1$, we have $\frac{Q(x)}{L(x)}|g(x)  $ and $ \frac{P(x)}{L(x)}|f(x)$.
    It follows that  $\deg (\frac{f(x)L(x)}{P(x)})< \deg(L(x))$ and
    \[
    \frac{v_if(a_i)}{P(a_i)}=  \frac{v_if(a_i)L(a_i)/ P(a_i)}{L(a_i)},
    \]
    for all $i=1,2,\dots,n-1$.
    Hence,  we have 
    $\boldsymbol{c}   \in {\rm  GRS}_\infty(\boldsymbol{a},L(x), \boldsymbol{v})$.  Therefore, 
    ${\rm  GRS}_\infty (\boldsymbol{a},P(x), \boldsymbol{v})\cap  {\rm  GRS}_\infty (\boldsymbol{a},Q(x), \boldsymbol{v})  \subseteq  {\rm  GRS}_\infty(\boldsymbol{a},L(x), \boldsymbol{v})$.
    
    By Lemma \ref{EPdivQ},  it follows that  $  {\rm  GRS}_\infty (\boldsymbol{a},L(x), \boldsymbol{v}) \subseteq   {\rm  GRS}_\infty(\boldsymbol{a},P(x), \boldsymbol{v})$ and  $  {\rm  GRS}_\infty(\boldsymbol{a},L(x), \boldsymbol{v}) \subseteq   {\rm  GRS}_\infty(\boldsymbol{a},Q(x), \boldsymbol{v})$.
    Hence,
    \[{\rm  GRS}_\infty(\boldsymbol{a},L(x), \boldsymbol{v}) \subseteq    {\rm  GRS}_\infty (\boldsymbol{a},P(x), \boldsymbol{v})\cap  {\rm  GRS}_\infty(\boldsymbol{a},Q(x), \boldsymbol{v}),\]
    and therefore, \[  {\rm  GRS}_\infty(\boldsymbol{a},P(x), \boldsymbol{v})\cap  {\rm  GRS}_\infty(\boldsymbol{a},Q(x), \boldsymbol{v}) =   {\rm  GRS}_\infty(\boldsymbol{a},L(x), \boldsymbol{v})\] as required.
\end{proof}

Using arguments similar to those in the proof of Corollary \ref{lcdPQ} and the properties given in the proof of Theorem  \ref{EPQ-intersection},
we have the following corollary.
\begin{corollary}
    Let $q$ be a prime power and $n\leq q+1$ be a positive integer.
    Further, let $P(x)$ and $Q(x)$ be polynomials in $\mathbb{F}_q[x]$ satisfying  the following conditions.
    \begin{enumerate}
        \item  ${\rm lcm}(P(x),Q(x))=M(x)$.
        \item $\gcd(P(x)Q(x), \prod\limits_{i=1}^n(x-a_i))=1$.
        \item $ \deg(M(x))\leq  n$.
    \end{enumerate}
    Then
    $ {\rm  GRS}_\infty(\boldsymbol{a},P(x), \boldsymbol{v})+  {\rm  GRS}_\infty(\boldsymbol{a},Q(x), \boldsymbol{v}) = {\rm  GRS}_\infty(\boldsymbol{a},M(x), \boldsymbol{v})$.
\end{corollary}

We now give necessary conditions for the existence of a linear $\ell$-intersection pair of MDS codes based on Theorems \ref{PQ-intersection} and \ref{EPQ-intersection}.

\begin{corollary}
    \label{corExistence}
    Let $q$ be a prime power and $n, k_1,k_2,\ell$ be non-negative integers such that $k_1\leq n\leq q+1$ and $k_2\leq n$.
    If there exist polynomials $P(x)$, $Q(x)$ and $L(x)$ in $\mathbb{F}_q[x]$ satisfying  the following conditions:
    \begin{enumerate}
        \item $\deg(P(x))=k_1$, $\deg(Q(x))=k_2$ and $\deg(L(x))=\ell$,
        \item  $\gcd(P(x),Q(x))=L(x)$,
        \item $\gcd(P(x)Q(x), \prod\limits_{i=1}^n(x-a_i))=1$,
        \item $\deg(P(x))+\deg(Q(x))\leq n+\deg(L(x))$,
    \end{enumerate}
    then there exists a linear $\ell$-intersection pair of MDS codes with parameters $[n,k_1,n-k_1+1]_q$  and  $[n,k_2,n-k_2+1]_q$.
\end{corollary}
\begin{proof}
    For $n\leq q$, GRS codes ${\rm  GRS}(\boldsymbol{a},P(x), \boldsymbol{v}$ and $ {\rm  GRS}(\boldsymbol{a},Q(x), \boldsymbol{v}) $ from a
    linear $\ell$-intersection pair of MDS codes with parameters $[n,k_1,n-k_1+1]_q$  and  $[n,k_2,n-k_2+1]_q$ by  Theorem \ref{PQ-intersection}.
    If $n\leq q+1$, then the extended GRS codes ${\rm  GRS}_\infty(\boldsymbol{a},P(x), \boldsymbol{v})$ and
    ${\rm  GRS}_\infty(\boldsymbol{a},Q(x), \boldsymbol{v})$ from a linear $\ell$-intersection pair of MDS codes with parameters
    $[n,k_1,n-k_1+1]_q$  and  $[n,k_2,n-k_2+1]_q$ by Theorem \ref{EPQ-intersection}.
\end{proof}

It is well known (see \cite[pp.  602-629]{G1981}) that the number of monic irreducible polynomials of degree $n$ over $\mathbb{F}_q$ is
\begin{align}
\label{numberOfMonic}
N_q(n)=\frac{1}{n} \sum_{d|n}  \mu(d) q^{n/d},
\end{align}
where  $\mu$ is the M\"obius function defined by
\[ \mu(m)=\begin{cases}
1 &\text{ if }m=1,\\
(-1)^r &\text{ if }m \text{ is a product of } r \text{ distinct primes},\\
0 &\text{ if } p^2|m \text{ for some prime }p.
\end{cases}\]
From \eqref{numberOfMonic}, $N_q(1)=q$,  $N_q(2)=\frac{q^2-q}{2}$ and
\begin{align*}
N_q(n) & \geq  \frac{1}{n}\left(q^n-  \sum_{d|n}    q^{n/d}  \right)\\
&\geq  \frac{1}{n}\left(q^n-   \sum_{d=0} ^{n-1}   q^{d}  \right)\\
&\geq  \frac{1}{n}\left(q^n-  \frac{(q^{n}-1)}{q-1}    \right)\\
&\geq  \frac{1}{n}\left( \frac{ (q-2)q^{n}+1}{q-1}            \right),
\end{align*}
for all  positive integers $n\geq 3$ since $\mu(1)=1$ and $\mu(d)\geq-1$ for all divisors $d$ of $n$.
Hence for $q\geq 3$, $N_q(1)=q\geq 3$,  $N_q(2)=\frac{q^2-q}{2}\geq q \geq 3$  and  $N_q(n) \geq \frac{q^{n-1}}{n} \geq q \geq 3$ for all $n\geq 3$.
Therefore,
\begin{align}
\label{Nqn}
N_q(n)    \geq 3,
\end{align}
for all prime powers $q\geq 3$ and positive integers $n$.

The following proposition guarantees that linear $\ell$-intersection pairs of MDS codes of length up to $q+1$ can be constructed for all possible parameters.
\begin{proposition}
    \label{prop-MDS-pair}
    Let $q\geq 3$ be a prime power and $n, k_1,k_2,\ell$ be non-negative integers such that $k_1\leq n\leq q+1$ and $k_2\leq n$. If $\ell \leq \min\{k_1,k_2\}$.
    Then  there exists a linear $\ell$-intersection pair of MDS codes with parameters $[n,k_1,n-k_1+1]_q$  and  $[n,k_2,n-k_2+1]_q$.
\end{proposition}
\begin{proof}
    Assume that $\ell \leq \min\{k_1,k_2\}$.
    By \eqref{Nqn}, there exist monic irreducible polynomials $f(x)$, $L(x)$,  and $h(x)$ in $\mathbb{F}_q$ of degrees $k_1-\ell$, $\ell$, and $k_2-\ell$, respectively,
    and the polynomial is set to $1$ if the degree is zero.
    Let $P(x)=f(x)L(x)$ and $Q(x)=h(x)L(x)$ so then
    $P(x)$, $Q(x)$ and $L(x)$ satisfy the conditions in Corollary \ref{corExistence}.
    Hence, $ {\rm  GRS}_\infty(\boldsymbol{a},P(x), \boldsymbol{v})$ and $ {\rm  GRS}_\infty(\boldsymbol{a},Q(x), \boldsymbol{v})$
    form  a linear $\ell$-intersection pair and  ${\rm  GRS}_\infty(\boldsymbol{a},P(x), \boldsymbol{v})$ and ${\rm  GRS}_\infty(\boldsymbol{a},Q(x), \boldsymbol{v})$
    have parameters $[n,k_1,n-k_1+1]_q$  and  $[n,k_2,n-k_2+1]_q$, respectively.
\end{proof}

Note that if $n\leq q$, ${\rm  GRS}(\boldsymbol{a},P(x), \boldsymbol{v})$ and ${\rm  GRS}(\boldsymbol{a},Q(x), \boldsymbol{v})$
form  a linear  $\ell$-intersection pair with parameters $[n,k_1,n-k_1+1]_q$  and  $[n,k_2,n-k_2+1]_q$,
where $P(x)$, $Q(x)$, and $L(x)$ are defined as in the proof of Proposition \ref{prop-MDS-pair}.


\subsection{Linear $\ell$-Intersection Pairs of  MDS Codes from Cauchy  and Vandermonde Matrices}
In this subsection, constructions of linear $\ell$-intersection pairs of MDS codes are given using super-regular matrices which are derived from some
subclasses of Cauchy and Vandermonde matrices.
A matrix $A$ over $\F_q$  is called {\em super-regular} if every square sub-matrix of $A$ is nonsingular (see \cite{RL}).
We recall  the following result.
\begin{proposition}[{\cite[Theorem 8]{Mac}}]
    \label{lem:Mac} 
    An $[n, k, d]_q$ code with generator matrix $G = [I \mid A]$ where $A$ is a $k \times (n - k)$ matrix
    is MDS if and only if $A$ is super-regular.  
    Equivalently, an $[n, k, d]_q$ code is MDS over $\F_q$  if and only if one of the following statements hold.
    \begin{enumerate}
        \item[(a)]  Every $k$ columns of a generator matrix are linearly independent. 
        \item[(b)] Every $n - k$ columns of a parity-check matrix are linearly independent.
    \end{enumerate}
\end{proposition}

The following theorem gives a construction of a linear $\ell$-intersection pair of MDS codes is given in terms of super-regular matrices.
\begin{theorem}
    \label{thm2: lintersctionMDS}
    Let $q$ be a prime power and $n$ be a positive integer.
    If there exists an $n\times n$ super-regular  matrix over $\mathbb{F}_q$, then there exist a linear $\ell$-intersection pair of MDS codes
    $[n,i]_q$ and $[n,j+\ell] _q$ for all $0\le i\le n$, $0\leq  j\leq n-i$ and  $0 \leq \ell  \le i$.
\end{theorem}
\begin{proof}
    Assume  that  there exists an   $n\times n$ super-regular  matrix $A$ over $\mathbb{F}_q$.
    Let $0\leq i\leq n$, $0\leq  j\leq n-i$ and  $0 \leq \ell  \le i$ be integers.
    Since $A$ is super-regular,  every square submatrix  of  $A$ is nonsingular.
    If $i=0$, let $C_0=\{\boldsymbol{0}\}$ be the zero code of length $n$  over $\mathbb{F}_q$, and
    otherwise let $A_i$ be  the  matrix formed by the first  $i$  rows of $A$.
    Then  the rows of $A_i$  are linearly independent and  $A_i$ generates  an $[n,i]_q$ linear code denoted by  $C_i$.
    Moreover, every $i$ columns of $A_i$ are linearly independent so by Proposition \ref{lem:Mac}(b), $C_i$ is an $[n,i]_q$ MDS code.
    If $j+\ell=0$, let $D_0=\{\boldsymbol{0}\}$ be the zero code of length $n$ over $\mathbb{F}_q$, and
    otherwise let $B_{j+\ell}$  be a matrix whose rows are composed of $\ell$ rows from $A_i$ and $j$ rows from the complement of $A_i$.
    Let $D_{j+\ell}$  be the linear code generated by  $B_{j+\ell}$.
    Using arguments similar to those in the proof of $C_i$, $D_{j+\ell}$ is an $[n, j+\ell ]$ MDS code.
    It is not difficult to see that  $\dim (D_{j+\ell} \cap  C_i)= \ell$  which implies  that the codes   $C_i$  and $D_{j+\ell}$  form
    a linear $\ell$-intersection pair of MDS codes with parameters $[n,i]_q$ and $[n,j+\ell] _q$.
\end{proof}

Note that if  $j=n-i$ in Theorem \ref{thm2: lintersctionMDS}, $C_i$  and $D_{n-i+\ell}$  form a linear $\ell$-intersection pair of  MDS codes 
$[n,i]_q$ and $[n,n-i+\ell] _q$  such that $C_i+D_{n-i+\ell}=\mathbb{F}_q^n$ for all $0\le i\le n$ and  $0\leq \ell \le i$.

The following corollary is an immediate consequence of Proposition \ref{lem:Mac} and  Theorem \ref{thm2: lintersctionMDS}.
\begin{corollary}
    \label{thm: lintersctionMDS}
    Let $q$ be  a prime power and $n$ be a positive integer.
    If there exists a systematic $[2n,n]_q$ MDS code over $\F_q$ with generator matrix $G=[I\mid A]$, then there exists a linear $\ell$-intersection
    pair of MDS codes $[n,i]_q$ and $[n,j+\ell] _q$ for all  $0\le i\le n$, $0\leq  j\leq n-i$ and  $0 \leq \ell  \le i$.
\end{corollary}

Based on Theorem \ref{thm2: lintersctionMDS}, linear $\ell$-intersection pairs of MDS codes can be constructed using Cauchy matrices.
Given $x_0, \ldots, x_{n-1}$ and $y_0, \ldots, y_{n-1}$ in $\mathbb{F}_q$, the  matrix
\begin{align}
\label{cauchymat}
A = [a_{ij} ]_{i= 0,1,\dots, n-1}^ {j=0,1,\dots,n-1},
\end{align}
where $a_{ij} = \frac{1}{x_i+y_j}$, is called a {\em Cauchy matrix}.
It is well known that
\[\
det(A) = \frac{\prod\limits _{ 0\le  i<j \le n-1} (x_j - x_i)(y_j-y_i)}{\prod\limits _{ 0 \le i,j\le  n-1} (x_i+y_j)}.
\]
Assuming $x_0, \ldots, x_{n-1}$  are distinct and $x_0, \ldots, x_{n-1}$ and $y_0, \ldots, y_{n-1}$ are distinct such that $x_i + y_j \neq 0$ for all $i, j$,
it follows that any square sub-matrix of a Cauchy matrix is nonsingular over $\mathbb{F}_q$.
In this case, the Cauchy matrix $A$ is super-regular.
From Theorem \ref{thm2: lintersctionMDS}, linear $\ell$-intersection pairs of MDS codes are obtained in the following corollary.

\begin{corollary}
    Let $q$ be a prime power and $n$ be a positive integer.
    If $n\leq \lfloor\frac{q}{2}\rfloor$, then there exists a linear $\ell$-intersection pair of MDS codes
    $[n,i]_q$ and $[n,j+\ell] _q$ for all $0\le i\le n$, $0\leq  j\leq n-i$ and $0 \leq \ell \le i$.
\end{corollary}

\begin{proof}
    Since $n\leq \lfloor\frac{q}{2}\rfloor$, we have $2n\leq q$.
    If $q$ is even, let $ x_0, x_1,\dots, x_{n-1}, y_0,y_1,\dots, $ $y_{n-1}$ be $2n$ distinct elements in $\mathbb{F}_q$.
    If $q$ is odd, let $ x_0, x_1,\dots, x_{n-1}$ be
    $n$ distinct non-zero elements in $\mathbb{F}_q$ such that $x_i\ne -x_j$ for all $0\leq i,j\leq n-1$
    and let  $y_i=x_i$ for all $i=0,1,\dots,n-1$.
    In both cases, it is not difficult to determine that $ x_0, x_1,\dots, x_{n-1}$ are distinct and $ y_0,y_1,\dots, y_{n-1}$
    are distinct such that $x_i+y_j\ne 0 $ for all $i,j$.
    From the discussion above, the $n\times n$ Cauchy matrix $A$ defined  in \eqref{cauchymat} is super-regular, and then the result follows from
    Theorem \ref{thm2: lintersctionMDS}.
\end{proof}

A {\em Vandermonde matrix} is an $n\times n$ matrix of the form  $V(a_1, a_2,\ldots, a_n)=
[{a_i}^{j-1}]_{i=1,\ldots,n}^{j=1,\dots,n}$, where $a_1, a_2,\ldots, a_n$ are elements of $\F_q$.
The determinant of  the  Vandermonde matrix $V(a_1, a_2,\ldots, a_n)$ is
\[
\det(V(a_1, a_2,\ldots, a_n))= \prod\limits _{1 \le, i \le j \le n}(a_j-a_i).
\]
The matrix $V(a_1, a_2,\ldots, a_n)$ is nonsingular if and only if all the $a_i$ are distinct.
In general, a Vandermonde matrix  is not super-regular.  
Vandermonde matrices can be used to construct linear $\ell$-intersection pairs of MDS codes since there exist super-singular matrices obtained from Vandermonde matrices.
Based on Theorem \ref{thm2: lintersctionMDS}, constructions of linear $\ell$-intersection pairs of MDS codes can be given using some subclasses of Vandermonde matrices.
A construction of super-regular matrices from Vandermonde matrices is given in the following proposition.
\begin{proposition}[{\cite[Theorem 2]{Lacan}}]
    \label{lacanT}
    Let $V(a_1,\ldots, a_n)$ and $V(b_1,\ldots, b_{n})$ be two Vandermonde matrices.
    Then the matrix  ${V(a_1,\ldots, a_n)}^{-1}V(b_1,\ldots, b_{n})$ is super-regular if and only if the $a_i$ and $b_j$ are $2n$ distinct elements.
\end{proposition}
Then by Theorem \ref{thm2: lintersctionMDS} and Proposition \ref{lacanT},
a construction of linear $\ell$-intersection pairs of MDS codes using Vandermonde matrices is given in the following corollary.
\begin{corollary}
    Let $q$ be a prime power and $n$ be a positive integer.
    If $n\leq \lfloor\frac{q}{2}\rfloor$, then there exists a linear $\ell$-intersection pair of MDS codes $[n,i]_q$ and $[n,j+\ell] _q$
    for all $0\le i\le n$, $0\leq  j\leq n-i$ and $0 \leq \ell  \le i$.
\end{corollary}

\begin{proof}
    Since $n\leq \lfloor\frac{q}{2}\rfloor$, let $ a_1, a_2,\dots, a_{n}$ and $ b_1, b_2,\dots, b_{n}$ be $2n$ distinct elements in $\mathbb{F}_q$.
    By Proposition \ref{lacanT}, the $n\times n$ matrix ${V(a_1,\ldots, a_n)}^{-1}V(b_1,\ldots, b_{n})$ is super-regular.
    The remainder of the proof follows from Theorem \ref{thm2: lintersctionMDS}.
\end{proof}

\section{Good  EAQECCs from Linear $\ell$-Intersection Pairs of Codes} \label{sectApp}
Entanglement-assisted quantum error correcting codes (EAQECCs) were introduced in Hsieh et al. \cite{hsieh} and can be constructed from
arbitrary classical codes. 
Further, the performance of the resulting quantum codes is determined by the performance of the underlying classical codes. Precisely, an $[[n,k,d;c]]_q$ EAQECC  encodes $k$ logical qudits into $n$ physical qudits using
$c$ copies of maximally entangled states and its  performance   is measured by its rate $\frac{k}{n}$ and net rate ($\frac{k-c}{n})$.
When the net rate of an EAQECC is positive it is possible to obtain catalytic codes as shown by Brun et al. \cite{brun2}. 
In \cite{GJG}, good entanglement-assisted quantum codes were constructed.
A link between the number of maximally shared qubits required to construct an EAQECC from a classical code and the hull of the classical code were given.  
For more details on EAQECCs, please refer to  \cite{brun}, \cite{GJG} and the references therein.

First, we recall the following useful proposition from  \cite{Wilde} which shows that EAQECCs can be constructed using classical linear codes.

\begin{proposition}[{\cite[Corollary 1]{Wilde}}]
    \label{prop:ent1}
    Let $H_1$ and $H_2$ be parity check matrices of two linear codes $D_1$ and $D_2$ with parameters $[n,k_1,d_1]_{q}$ and $[n,k_2,d_2]_{q}$, respectively.
    Then an $[[ n,k_1 +k_2-n+c, \min \{d_1,d_2\};c ]]_q$ EAQECC can be obtained where $c=\rank (H_1H_2{^{t}})$ is the required number of maximally entangled states.
\end{proposition}

\subsection{EAQECCs from Linear $\ell$-Intersection Pairs of Codes}

Linear $\ell$-intersection pairs of codes can be used to construct EAQECCs based on Proposition \ref{prop:ent1} as follows.
\begin{proposition}
    \label{EAQECC-intersection-pair}
    Let $\ell\geq 0$  be an integer and $C_1$ and $C_2 $ be a linear $\ell$-intersection pair of codes with parameters $[n,k_1,d_1]_q$  and  $[n,k_2,d_2]_q$, respectively.
    Then there exists an $[[ n, k_2-\ell , \min \{d_1^\perp,d_2 \};k_1-\ell ]]_q$ EAQECC with $d_1^\perp=d(C_1^\perp)$.
\end{proposition}
\begin{proof}
    If $D_1=C_1^\perp$ and $D_2=C_2$ in Proposition \ref{prop:ent1}, then the result follows from  Proposition \ref{prop:ent1} and Theorem \ref{thmGLCP}.
\end{proof}

\begin{corollary} \label{thm:intersect}  Let $n$ and $r$  be positive integers such that $r< \frac{n}{2}$.  Let $ k_1 $   and $k_2$ be integers such that $r\leq k_1< n-r \leq k_2\leq n $.  
    If there exists an $[n,k_2,d]_q$ code, then there exists a positive net rate  $[[ n, k_2-\ell ,  d ;k_1-\ell ]]_q$ EAQECC $Q$ for some $0\leq \ell \leq k_1$.  In addition, if $\ell\leq \frac{n}{2}-r$,  the rate of EAQECC $Q$ is greater than or equal to $\frac{1}{2}$. 
\end{corollary}
\begin{proof}   Assume that there exists    an $[n,k_2,d]_q$ code, denoted by $C_2$. Since $ n-k_1\leq k_2$,  there exisits a  linear code   $D$  with parameters $[n,n-k_1,d_1^\perp]_q$ and $d_1^\perp\geq d$. Let  $C_1=D^\perp$.
    Then $C_1$ and $C_2$ form  a linear $\ell$-intersection pair of $[n,k_1]_q$ and  $[n,k_2,d]_q$  for some $0\leq \ell \leq k_1$  and $d(C_1^\perp)=d_1^\perp \geq d$.  By  Proposition  \ref{EAQECC-intersection-pair},     there exists an $[[ n, k-\ell , \min\{d_1^\perp,d\};k_1-\ell ]]_q=[[ n, k-\ell , d;k_1-\ell ]]_q$ EAQECC $Q$. Consequently, the net rate of $Q$ is \[\frac{(k_2-\ell)-(k_1-\ell)}{n}=\frac{k_2-k_1}{n}>0.\] 
    
    In addition, assume that $\ell\leq \frac{n}{2}-r$.  Then  the rate of $Q$ is \[ \frac{k_2-\ell}{n}\geq  \frac{(n-r)- (n/2-r)}{n}   = \frac{1}{2}\]
    as desired.
\end{proof}

To obtain an EAQECC with good minimum distances, the input linear code in Corollary \ref{thm:intersect}  can  be chosen from the best-known linear codes in \cite{G2019} or in the database of \cite{BCP1997}. Moreover, the required number of maximally entangled states $c=k_1-\ell$  can be adjusted using a weighted permutation matrix as in Lemma~\ref{rem1} and Example~\ref{exwp}.

Using the  arguments in MAGMA shown below, it can be  easily seen that a large number of  linear $\ell$-intersection pairs of best-known linear codes  constructed as in the proof of Corollary \ref{thm:intersect} satisfy  the condition $\ell\leq \frac{n}{2}-r$.   Consequently, many EAQECCs obtained in Corollary \ref{thm:intersect}   are good in the sense that  they have good rate and positive net rate.    
To save spaces, the computational results are omitted.

\begin{center}
    \rule{0.7\textwidth}{.4pt}
    
    \texttt{ 
        \begin{tabular}{l}
            q:= (the cardinality of the finite field);\\
            a:= (the starting point for the length);\\
            b:= (the end point for the length);\\
            for n in [a..b] do\\
            ~ for r in [1..Floor((n-1)/2)] do\\
            ~~ for k2 in [n-r..n] do\\
            ~~~ for k1 in [r..n-r-1]  do\\
            ~~~~ C2:=BKLC(GF(q),n,k2);\\
            ~~~~ C1:=Dual(BKLC(GF(q),n,n-k1));\\
            ~~~~ l:=Dimension(C1 meet C2);  \\
            ~~~~ d:=MinimumDistance(C2);\\
            ~~~~ if  l le  n/2-r then \\
            ~~~~~ "[[",n,k2-l, d, k1-l,"]]";\\
            ~~~~ end if; \\
            ~~~ end for;\\
            ~~ end for;\\
            ~ end for;\\
            end for;
    \end{tabular}}
    
    \rule{0.7\textwidth}{.4pt}
\end{center}

By Theorem \ref{thmGLCP},  the statement  ``\texttt{l := k1-rank(G1*Transpose(H2));}'' can be replaced by   ``\texttt{G1:=  GeneratorMatrix(C1); H2:= ParityCheckMatrix(H2); l := k1-rank(G1*Transpose(H2));}''.

\subsection{MDS EAQECCs from Linear $\ell$-Intersection Pairs of Codes}

In this subsection, we focus on maximum distance separable (MDS)  EAQECCs derived from linear $\ell$-intersection pairs of MDS linear codes.

The Singleton bound for an EAQECC is given in the following proposition.
\begin{proposition}[{\cite{brun}}]
    \label{Bound:Singleton}
    An $[[n,k,d;c]]_q$ EAQECC satisfies
    \[
    n +c -k \ge 2 (d-1),
    \]
    where $0 \le c \le n-1$.
\end{proposition}
An EAQECC attaining this bound is called an {\em MDS EAQECC}.
Some constructions of MDS EAQECCs were given in  \cite{CHFC2017}, \cite{GJG}, \cite{LLGML2018}, \cite{LMLMLC2018}  and  \cite{QZ2018}.

Here, MDS EAQECCs are constructed  from linear $\ell$-intersection pairs of MDS codes based on  Proposition \ref{EAQECC-intersection-pair} as follows.
\begin{proposition}
    Let $\ell\geq 0$ be an integer and $C_1$  and $C_2 $ be a linear $\ell$-intersection pair of MDS codes with parameters $[n,k_1,n-k_1+1]_q$  and  $[n,k_2,n-k_2+1]_q$, respectively.
    Then there exists an $[[ n, k_2-\ell , \min\{ k_1+1, n-k_2+1\};k_1-\ell ]]_q$ EAQECC.
    In particular,  if $n=k_1+k_2$, then there exists an $[[ n, k_2-\ell , k_1+1;k_1-\ell ]]_q$ MDS EAQECC.
\end{proposition}

\begin{proof}
    Since $C_1$ is MDS with parameters $[n,k_1,n-k_1+1]_q$, it follows that $C_1^\perp$ is also MDS with parameters $[n,n-k_1,k_1+1]_q$.
    Then an $[[ n, k_2-\ell , \min\{ k_1+1, n-k_2+1\};k_1-\ell ]]_q$ EAQECC exists by Proposition \ref{EAQECC-intersection-pair}.
    Now assume that $n=k_1+k_2$.
    Then $k_1+1= n-k_2+1$ and $n + k_1-\ell - (k_2-\ell) = n+k_1-k_2 = 2 k_1$ and hence the EAQECC given above is MDS.
\end{proof}

From Proposition \ref{prop-MDS-pair}, their exists a linear $\ell$-intersection pair of MDS codes with parameters
$[n,k_1,n-k_1+1]_q$  and  $[n,k_2,n-k_2+1]_q$ for all $q\geq 3$, $k_1\leq n\leq q+1$ and $k_1\leq n$ such that $\ell \leq \min\{k_1,k_2\}$.
Then we have the following corollary.
\begin{corollary}
    \label{cor-construction}
    Let $q\geq 3$ be a prime power and $n, k_1,k_2,\ell$ be non-negative integers such that $k_1\leq n\leq q+1$ and $k_1\leq n$. If $\ell \leq \min\{k_1,k_2\}$.
    Then there exists an $[[ n, k_2-\ell , \min\{ k_1+1, n-k_2+1\};k_1-\ell ]]_q$ EAQECC.
\end{corollary}

By setting $k=k_1=n-k_2$ in Corollary \ref{cor-construction}, we have the following construction for MDS EAQECCs.
\begin{corollary}
    \label{cor:4.6}
    Let $q \geq 3$ be a prime power and $0\leq k \leq n\leq q+1$ be integers.
    Then there exists an $[[ n, n-k-\ell , k+1;k-\ell ]]_q$ MDS EAQECC for all integers $0\leq \ell \leq \min\{k,n-k\}$.
\end{corollary}

In the related literature, e.g. \cite{CHFC2017},  \cite{GJG},  \cite{LLGML2018}, \cite{LMLMLC2018} and  \cite{QZ2018},
the required number of maximally entangled states of most MDS EAQECCs is fixed.
Recently in \cite{LC2018}, MDS EAQECCs were given where the required number of maximally entangled states varies according to
$Hull(C):=C\cap C^\perp$ of a classical code  $C$. 
However, for an odd prime power $q$, there are many restrictions on the length $n$ of MDS EAQECCs (see \cite [Theorems 17 and 18]{LC2018}).
Based on the construction given in Corollary \ref{cor:4.6}, the length $n$ can be arbitrary from $1$ to $q+1$.
Hence, in many cases, the parameters of the MDS EAQECCs in Corollary \ref{cor:4.6}  are new.

\section{Conclusion}
\label{secConclud}
In this paper,  a  linear  $\ell$-intersection pair of codes was introduced as a
generalization of linear complementary pairs of codes. 
Further, characterizations and constructions of such pairs of codes were given  in terms of the corresponding generator and parity-check matrices.
Linear  $\ell$-intersection pairs of MDS codes over $\mathbb{F}_q$ of length up  to $q+1$ were given for all possible parameters. 
As an application, linear  $\ell$-intersection pairs of codes were employed to construct entanglement-assisted quantum error correcting codes.
An interesting research problem for the future is to investigate the applications of these pairs of codes to secure communications against attacks. Further bounds on the parameters of these codes would be  of great interest. Constructions of linear  $\ell$-intersection pairs of good codes with prescribed value  $\ell$  and  solving Conjecture \ref{cojecture} are interesting problems as well.     

\section*{Acknowledgments}  The authors would like to thank the anonymous referees for very helpful comments.    S. Jitman  was supported by the Thailand Research Fund and  Silpakorn University under Research Grant RSA6280042.

\end{document}